\documentclass[sigconf]{acmart}
\usepackage{color}
\usepackage{hyperref}
\usepackage{balance}
\usepackage[ruled,vlined,linesnumbered]{algorithm2e}
\newcommand{\dis}{{\sc DDP}}

\begin{document}
	
	\sloppy
	\title{Approximation Algorithms for Drone Delivery Packing Problem}
	
	\author{Saswata Jana}
	\affiliation{%
		\institution{Indian Institute of Technology Guwahati}
		%\streetaddress{P.O. Box 1212}
		\city{Guwahati, India} 
	}
	\email{saswatajana@iitg.ac.in}

	\author{Partha Sarathi Mandal}
	\affiliation{%
		\institution{Indian Institute of Technology Guwahati}
		%\streetaddress{P.O. Box 1212}
		\city{Guwahati, India} 
	}
	\email{psm@iitg.ac.in}

	\begin{abstract}
		Recent advancements in unmanned aerial vehicles, also known as drones, have motivated logistics to use drones for multiple operations. Collaboration between drones and trucks in a last-mile delivery system has numerous benefits and reduces a number of challenges. In this paper, we introduce  \textit{drone-delivery packing problem} (DDP), where we have a set of deliveries and respective customers with their prescribed locations, delivery time intervals, associated cost for deliveries, and set of drones with identical battery budgets. The objective of the DDP is to find an assignment for all deliveries to the drones by using the minimum number of drones subject to the battery budget and compatibility of the assignment of each drone.
		We prove that DDP is NP-Hard and formulate the integer linear programming (ILP) formulation for it. 
		We proposed two greedy approximation algorithms for DDP. The first algorithm uses at most $2OPT + (\Delta + 1)$ drones. The second algorithm uses at most $2OPT + \omega$ drones, where OPT is the optimum solution for DDP, $\omega$ is the maximum clique size, and $\Delta$ is the maximum degree of the interval graph $G$ constructed from the delivery time intervals. 
	\end{abstract}
	
		% <concept_id>10010147.10010178.10010219.10010220</concept_id>
		% <concept_desc>Computing methodologies~Multi-agent systems</concept_desc>
		% <concept_id>10010147.10010178.10010219.10010222</concept_id>
		% <concept_desc>Computing methodologies~Mobile agents</concept_desc>
		% <concept_id>10010147.10010178.10010219.10010221</concept_id>
		% <concept_desc>Computing methodologies~Intelligent agents</concept_desc>
	
	
	
	\maketitle
	
	\fancyhead{}
	
	\section{Introduction}

	% Delivery with Drones in a region is a problem of significant interest in recent time \cite{}. Recently, this problem has been formulated by Augustine and Moses Jr.~\cite{Augustine:2018} in the context of graphs as follows: 
	
	% The dispersion of autonomous mobile robots in a region
	% % to spread them out evenly in a region 
	% is a problem of significant interest in distributed robotics \cite{HsiangABFM02,Hsiang:2003}. 
	% Recently, this problem has been formulated by Augustine and Moses Jr.~\cite{Augustine:2018} in the context of graphs  % (instead of a region).  %Augustine and Moses Jr.~\cite{Augustine:2018} 
	% %The problem is defined 

	The rapid demand for commercial deliveries motivates the e-commerce giants to deliver more effectively to customers.
	%The rapid demand in commercial deliveries motivates the e-commerce giants to make more effective deliveries to the customers. 
	% In this delivery journey \textit{last-mile delivery} \cite{LastMile}, is the most expensive and time consuming process, where the product need to deliver to the customer doorstep from the distribution hub. This delivery needs lot of human interactions. But, advances in drone technologies makes a miniature and this delivery system brings more impact in today's pandemic world. Big companies started their preparation for making productive parcel deliveries through drones \cite{amazon}. In addition, using key vehicles (like trucks, cars, vans, etc.) along with drones with their constraints enhances the profit and diminishes the total delivery time interval.
	% Furthermore, drones have enormous application in the field of defence and disaster response \cite{defence}, agriculture \cite{agriculture}, healthcare \cite{healthcare}, etc.
	This delivery journey \textit{last-mile delivery} \cite{LastMile}, is the most expensive and time-consuming process, where the product needs to deliver to the customer's doorstep from the distribution hub. This delivery requires a lot of human interactions. But, advances in drone technologies make a miniature, and this delivery system brings more impact in today's pandemic world. Big companies started preparing to make productive parcel deliveries through drones \cite{amazon}. In addition, using key vehicles (like trucks, cars, vans, etc.) along with drones with their constraints enhances the profit and diminishes the total delivery time interval.
	Further drones have enormous application in the field of defence and disaster response \cite{defence}, agriculture \cite{agriculture}, healthcare \cite{healthcare}, etc.
	
	\vspace{1mm}
	\noindent \textit{Challenges:}
	% Delivery of the packages by the drones in association with a truck creates a lot of challenges in our real-world scenario. For a given location of customers, we need to know the optimum route for the truck as well as the launching and rendezvous points of the drones and the truck for making multiple deliveries. Also, the limited battery budget of the available drones in the market
	% do not allow us to make a desirable number of deliveries. For making all the delivery we need to either create rechargeable stations or battery replacement policy for the drones. This takes lot of completion. 
	% In addition, they can't ship a drone for any set of deliveries altogether because of conflicts among delivery time intervals. At a time, it can deliver at most one package.
	% All these complexities influence logistics to complete all the deliveries by using few identical drones so that total delivery cost by the company is optimized.
	Delivery of the packages by drones in association with a truck creates many challenges in our real-world scenario. For a given location of customers, we need to know the optimum route for the truck, the launching and rendezvous points of the drones, and the truck for making multiple deliveries. Also, the limited battery budget of the available drones in the market does not allow us to make a desirable number of deliveries. To guarantee all the deliveries for a given delivery time with a fixed number of drones, we need to either create rechargeable stations or a battery replacement policy for the drones, which makes the system complicated.  
	In addition to this, we can't use a drone for any set of deliveries altogether because of conflicts among delivery time intervals. 
	At a time, a drone can deliver at most one package.
	All these complexities influence logistics to complete all the deliveries by using a few identical drones so that the total delivery cost by the company is optimized.

	\vspace{1mm}
	\noindent \textit{Drone delivery
		scheduling problem:}
	This problem considers last-mile delivery to customers 
	using a set of drones carried by a truck moving in a prescribed route. For a given set of deliveries and their delivery time intervals, reward or profit for each delivery, and battery budget of the drones, the goal is to schedule a given/ fixed set of drones for the deliveries to make total profit maximum. This problem was introduced by Sorbelli et al. \cite{Betti_ICDCN22} and proved that the problem is NP-Hard and proposed heuristic algorithms to solve the problems. 
	
	The aforementioned problem does not guarantee delivery of all deliveries because it uses a fixed set of drones to optimize total profit for delivery.
	In this paper, we introduce the Drone-Delivery Packing Problem (DDP), which aims to optimize the number of drones to deliver all deliveries. We propose two approximation algorithms to solve DDP.

	\vspace{1mm}
	\noindent{\bf Contributions.}
	In this paper our contributions are following.
	\begin{itemize}
		\item In this paper, we introduce Drone Delivery Packing Problem (DDP) and prove that it is NP-Hard. We present an integer linear programming formulation for it, which is only suitable for solving the problem optimally for small-sized instances.
		\item We propose an approximation algorithm for DDP with running time $\mathcal{O}(n\log n + n_e)$ and uses at most $2OPT + (\Delta + 1)$ drones, where $n$ is the number of deliveries, $n_e$ is the total number of edges and $\Delta$ is the maximum degree of the interval graph $G$.
		\item We propose another approximation algorithm for DDP with running time $\mathcal{O}(n\log n + n_e)$ and uses at most $2OPT + \omega$ drones, where $n$ is the number of deliveries, $n_e$ is the total number of edges and $\omega$ is the maximum clique size of the interval graph $G$.
	\end{itemize}

	\vspace{1mm}
	\noindent{\bf Related Work.}
	Since the drones have certain mobility, using the truck makes more efficient deliveries. Various researches have been done in this area of delivery with collaboration between drones and a truck.

	This kind of delivery comes into account when \textit{Murray and Chu} in \cite{MURRAY201586} introduced \textit{flying sidekicks traveling salesman problem}, a more extension of TSP, where customers need to visit (or deliver) either by the truck or by a drone starting from the depot. At the same time, the drone begins its journey either from the depot or any customer location, and the same for the meeting occasion. Here authors want to minimize the total makespan to make all the deliveries. For this purpose, they proposed an optimal \textit{mixed integer linear programming} (MILP) formulation and two effective heuristic solutions.
	
	\textit{Crisan and Nechita} \cite{CRISAN201938} proposed another effective heuristic for \textit{flying sidekicks traveling salesman problem} by using the solution for TSP. \textit{Murray and Raj} \cite{Murray2019TheMF} extended \textit{flying sidekicks traveling salesman problem} for multiple drones. Here they proposed MILP formulation for the problem and then a heuristic solution with numerical testings. 
	\textit{Daknama and Kraus} \cite{Daknama2017VehicleRW} take in the hand of mobility of drones. There is a rechargeable area on the truck's roof where the drones can charge after completing one delivery and go for the next delivery. Here authors proposed a heuristic algorithm for scheduling of truck and drones.
	Delivery by drones only in the attention in \cite{Boysen2018DroneDF}, where \textit{Boysen et al.} objective to find the launch and meet point for delivery with the truck so that total makespan for completing all the deliveries minimized. The assumption for this problem was the knowledge of the truck's route but without any battery constraint of the drones. 
	Mathew et al. \cite{HDP} proposed a \textit{heterogeneous delivery problem}, where two co-operative vehicles (truck and \textit{micro-aerial vehicles} (MAV)) are used for performing all the deliveries. Whereas drone can fly from and meet with the truck at the prescribed warehouses. Their goal is to find the optimal route with respect to the cost. For this, they proposed two heuristic solutions and the hardness of the problem.
	
	Very recently, \textit{Sorbelli et al.} \cite{Betti_ICDCN22} proposed a
	Multiple Drone-Delivery Scheduling
	Problem (MDSP), where a truck and multiple drones cooperate among them-self for package delivery in a last-mile.
	The paper gave NP-hardness proof of the problem, ILP formulation, and design a heuristic algorithm for the single drone case and two heuristic algorithms for the multiple drones case.
	
	The problem we discuss here is a more generalized version of bin-packing problem. \textit{Coffman et al.} in \cite{bin_survey} presents several variants of bin-packing prolem with their approximation algorithm. \textit{Stacho} in \cite{chordal} described various colouring of chordal graph along with their complexities.

	\vspace{1mm}
	\noindent{\bf Roadmap.} We discuss 
	model and preliminaries with problem definition in Section \ref{section:model}. 
	We present hardness of the Drone-Delivery Packing Problem and ILP formulation of it in Section \ref{section:hardness}. 
	We discuss propose two approximations in Section \ref{section:aproxAlgos} 
	Finally, we conclude in Section \ref{section:concl}.  
	
	\section{Model and Preliminaries}
	\label{section:model}
	
	\noindent \textbf{Model:} Let $\mathcal{N} = \{1, 2, \cdots, n\}$ be the set of \textit{deliveries} to be delivered to the respective customers with their prescribed location at $\delta_j$ for each $j \in \mathcal{N}$. A delivery company wants to deliver the packages to the corresponding customers  by the drones having identical battery capacity or budget $B$. Let $c_j$ $(0 < c_j \leq B)$ be the energy cost for a drone to complete the delivery $j \in \mathcal{N}$. Initially, all the drones are at the company's warehouse (depot). Now, a truck will leave the depot with all the drones in its pre-decided path. For making a delivery $j$ at position $\delta_j$, a drone will fly from the truck at a specific \textit{launching location} $(\delta_j^L)$ and after completing the delivery at $\delta_j$ it meets with the truck again at a specific \textit{rendezvous location} $(\delta_j^R)$. After completing all the deliveries, the truck with the drones returns to the depot.
	
	Let at time $t_0$ the truck start its journey with all the drones. 
	For each delivery $j$, at the time $t_j^L$ and $t_j^R$ one drone comes at the position $\delta_j^L$ and $\delta_j^R$, respectively. Then, $I_j = [t_j^L, t_j^R]$ be the \textit{delivery time interval} for the delivery $j \in \mathcal{N}$. Let $I = \{I_1, I_2, \cdots, I_n\}$ be the \textit{delivery time interval set} for the set of deliveries $\mathcal{N}$. Also, emphasize that the truck will only move in one direction. So, if $A$ and $B$ are any two points on the truck's pre-decided path, where $B$ is the later point of $A$, then $t_A < t_B$, where $t_A$ and $t_B$ are the time when the truck arrives at the position $A$ and $B$, respectively.
	
	Any drone can be assigned for multiple deliveries $S$ ($S \subseteq I)$ constraints to the battery budget $B$ of the drone and the compatibility of the delivery time intervals. Any two delivery time intervals $I_j$ and $I_k$ are said to be \textit{compatible} or \textit{conflict-free} if $I_j \cap I_k = \phi$, otherwise they are in \textit{conflict}. Any set of delivery time intervals $S \subseteq I$ is said to be \textit{compatible} if all pairs of intervals in it are compatible. A compatible set $S \subseteq I$ is said to be feasible if $\sum_{I_j \in S} c_j \leq B$. A feasible set $S \subseteq I$ can be assigned or packed to a drone. If a feasible set $S \subseteq I$ is assigned for the  drone $i$ then, we call $S$ as an \textit{assignment} of the drone $i$. We are saying a drone is in \textit{used} or \textit{opened} if there exists at least one delivery which is assigned to this drone.
	
	\vspace{1mm}
	\noindent{\bf Graph.} For the given set of delivery time intervals $I$, we can construct an interval graph $G=(V,E)$, where the vertices represent the intervals and two vertices are adjacent if the corresponding two intervals conflict.  $\Delta$ is denoted as the maximum degree of $G$, $|V| = |I|=n$ and $|E|=n_e$.
	Now $G =(V,E)$ being an interval graph, it is \textit{perfect} \cite{west}. Therefore, $\chi(G)$ = $\omega(G)$, where $\chi(G)$ is the \textit{chromatic number} and $\omega(G)$ is the \textit{maximum clique size} of $G$.
	We use $\omega$ instead of $\omega(G)$ for simplicity. So, the vertices of $G$ can be optimally coloured by using $w$ many colours in polynomial time (linear in terms of number of vertices and edges) \cite{chordal} such that no two adjacent vertices get the same colour. Since every interval in $I$ represented by a unique vertex in $G$, each interval can be coloured linearly from the colour of $G$.  Further, we can partition $I$ into $\omega$ many conflict-free (i.e., compatible) sets corresponding to each colour. Let $J_1, J_2, \cdots, J_{\omega}$ be the partitions. Therefore, $\cup_{k=1}^{\omega} J_k = I$.
	
	% Now, for each $k$ $(1 \leq k \leq \omega)$, we can find the number of drones needed to schedule all the intervals in $J_k$ and the corresponding schedules using algorithm \ref{GreedyAlgo_mod}, a modified version of algorithm \ref{GreedyAlgo}. Similar to the algorithm \ref{GreedyAlgo}, for each interval set $J_k$, we construct a balanced binary search tree $T_k$ $(1 \leq k \leq \omega)$.  

	\vspace{1mm}
	\noindent{\bf Drone-Delivery Packing Problem (DDP).} {\dis} is formally defined as follows.
	\begin{definition} [{\dis}] \label{Problem 1}
		{\em 
			Given a set of delivery time intervals $I = \{I_1, I_2, \cdots, I_n\}$ corresponding to the set of deliveries $\mathcal{N}$ associated with cost $c_j$ for each $j \in \mathcal{N}$, the objective for DDP is to use minimum number of drones such that each delivery is completed by exactly one drone and each assignment of the opened drones is feasible.}
	\end{definition}
	In another way, find a smallest cardinality set of drones $M^* = \{1, 2, \cdots, m^*\}$ along with a family of assignments $\mathcal{S^*} = \{S_1^*, S_2^*, \cdots, S_{m^*}^*\}$, where $S^*_i \subseteq I$ is a feasible assignment for the drone $i \in M^*$ such that each delivery is associated with a unique element (assignment) in $\mathcal{S}^*$ and $S^*_l \cap S^*_m = \phi$; $\forall  1\leq l \neq m \leq m^*$.
	
	\section{Problem Hardness} \label{section:hardness}
	Here we establish the hardness of the DDP via the \textit{bin packing problem} \cite{Coffman1984}.
	\begin{theorem}
		\label{Np-hard}
		DDP is an NP-hard problem.
	\end{theorem}
	\begin{proof}
		We prove it by the reduction from bin-packing problem (BP).
		
		The goal of the bin-packing problem is to pack given a set of items associated with some sizes into a set of bins with identical capacity such that the number of bins used for the packing is minimum.
		
		Let $\mathcal{I}_{BP} = (\mathcal{N}_{BP}, s, b)$ be an arbitrary instance of BP, where $\mathcal{N}_{BP} = \{1, 2, \cdots,$ $ n\}$ is the set of items, $b$ is the capacity of each identical bin and  $s \colon \mathcal{N}_{BP} \to (0,b]$ is the size function of the items with $s(j) = s_j, \forall j \in \mathcal{N}_{BP}$.
		
		The above instance $(\mathcal{I}_{BP})$ for BP can be transformed into an instance of DDP as follows. Let $\mathcal{I} = (\mathcal{N}, I, c, b)$, where $\mathcal{N} = \mathcal{N}_{BP}$ is the set of deliveries, $I = \{I_1, I_2, \cdots, I_n\}$ is the set of delivery time intervals with $I_j = [2j, 2j+1]$ being the delivery time interval of the delivery $j$ $(1 \leq j \leq n)$, $B = b$ is the battery budget of identical drones and $c \colon \mathcal{N} \to (0,B]$ is the cost function of the deliveries with $c(j) = c_j = s_j, \forall j \in \mathcal{N}$. $\mathcal{I}$ is an instance of DDP and the reduction from $\mathcal{I}_{BP}$ is polynomial.
		
		All the delivery time intervals in $I$ are pairwise compatible. So, if all the items in $\mathcal{N}_{BP}$ can be packed using $m$ bins, then all the deliveries in $\mathcal{N}$ can be scheduled using $m$ drones and vice-versa. Thus, BP $\leq_P$ DDP. Hence, being BP an NP-hard problem \cite{Np}, DDP is so.
	\end{proof}
	\subsection{ILP Formulation}
	For solving DDP optimally, we can formulate the problem via \textit{integer linear programming} (ILP) as follows. According to the assumption, the associated cost $c_j$ for each delivery $j \in \mathcal{N}$ does not exceed the battery budget of the drones. So, $n = |\mathcal{N}|$ is an upper bound for the optimal solution of the DDP. Let $\mathcal{M} = \{1, 2, \cdots, n\}$ be the set of drones available in the company's warehouse.
	Let,
	\begin{equation}
		\begin{split}
			y_{i} & = 1 \text{, if drone $i \in \mathcal{M}$ is used} \\
			& = 0 \text{, otherwise}.
		\end{split}
	\end{equation}
	\begin{equation}
		\begin{split}
			x_{ij} & = 1 \text{, if delivery $j$ $\in$ $\mathcal{N}$ completed by the drone $i$ $\in$ $\mathcal{M}$} \\
			& = 0 \text{, otherwise}.
		\end{split}
	\end{equation}
	\begin{equation}
		\min \sum_{i \in \mathcal{M}}^{} y_i \label{eq:2}
	\end{equation}
	%subject to
	\begin{equation}
		\text{subject to}	\sum_{j \in \mathcal{N}}^{}  c_{j} x_{ij} \leq B.y_i , \text{    \hspace{85pt}     $\forall$ $i$ $\in$ $\mathcal{M}$} \label{eq:3}
	\end{equation}
	\begin{equation}
		\sum_{i \in \mathcal{M}}^{} x_{ij} = 1, \text{ \hspace{80pt}   $\forall$ j $\in$ $\mathcal{N}$} \label{eq:4}
	\end{equation}
	\begin{equation}
		x_{ij} + x_{ik} \leq 1, \text{ \hspace{20pt}  $\forall i \in \mathcal{M}$; $\forall j, k \in \mathcal{N}$ such that $I_{j} \cap I_{k} \neq \phi$} \label{eq:5}
	\end{equation}
	\begin{equation}
		y_i, x_{ij} \in \{0,1\}, \text{ \hspace{80pt}    $\forall i \in \mathcal{M}$, $\forall j \in \mathcal{N}$}
	\end{equation}
	The objective function (\ref{eq:2}) is about minimizing the number of used drones. Constraint (\ref{eq:3}) depicts that a used drone cannot have an assignment which exceeds it's battery budget $B$. Constraint (\ref{eq:4}) says that a customer can be delivered by exactly one drone. Constraint (\ref{eq:5}) tells that if two deliveries are in conflict, a drone can choose at most one of them.
	
	The aforementioned formulation is only suitable
	for solving the problem optimally for small-sized instances.

	\section{Approximation Algorithms} \label{section:aproxAlgos}
	In this section, we design two approximation algorithms to solve the  DDP. First, we propose a greedy approximation algorithm (\textsc{GreedyAlgoForDDP}) and then we propose another approximation algorithm using colouring (\textsc{ApprxAlgoWithColouring}) that uses \textsc{GreedyAlgoModified} as a subroutine.  
	\textsc{GreedyAlgoModified} is 
	a modified version of \textsc{GreedyAlgoForDDP}.
	
	\subsection{Greedy Approximation Algorithm for DDP}
	Algorithm \ref{GreedyAlgo} takes a simple greedy approach like First-Fit bin packing algorithm \cite{Coffman1984}.  At first, the algorithm sorts the deliveries according to their launching time, takes the delivery with the shortest launch time, and assigns it to any of the drones. The assigned drone is termed as \textit{used} or \textit{opened} drone. Then it takes the delivery one by one as per the sorted order and tries to assign them to the used drones. 
	%A new drone is introduced if any delivery does not fit into any of the used drones.
	A new drone is introduced if any delivery does not fit, i.e., not feasible with any of the used drones.
	
	For getting a better time complexity of the above greedy approach, we are using a balanced binary search tree with each node corresponding to a used drone, and the key of that node is the remaining battery capacity of the used drone. Furthermore, for each node, we store an additional data corresponding to the maximum time among all rendezvous times of the associated drone's assigned deliveries. All other attributes of the node are specified below. With the help of this tree, we can efficiently either find the index of the used drone in which we can assign a particular delivery (in addition to the existing assignment) or get the confirmation to introduce a new drone. The pseudocode of this greedy approach with the tree depicts in Algorithm \ref{GreedyAlgo}.
	
	\vspace{1mm}
	\noindent \textbf{Variable Specification:} $T$: A balanced binary search tree $T$ is represented by \textit{root} node (initially $root$ = NULL).
	\textit{node/ NODE}: Tree nodes with five attributes. For $node$ we use \textit{node.index}: index of the corresponding drone; \textit{node.key}: remaining capacity of the drone; \textit{node.data}: maximum among the rendezvous time of all assigned deliveries to the drone; \textit{node.left}: left child of the node in the tree; \textit{node.right}: right child of the node in the tree. Similarly, we can consider the attributes; \textit{NODE.index}, \textit{NODE.key}, \textit{NODE.data},  \textit{NODE.left}, \textit{NODE.right} for $NODE$. 
	The tree $T$ maintains it's balance and search property according to node's \textit{key} attribute.
	
	\vspace{1mm}
	\noindent \textbf{Global variable:} $d_j = N(j) + 1$: Maximum number of checking need to be done in the tree for the delivery $j$ to maintain compatiblity property of an assignment, whereas $N(j)$ is the number of conflicts for the delivery $j$, $\forall j: 1 \leq j \leq n$.
	
	\noindent \textit{Description of Algorithm \ref{GreedyAlgo}:} Initially, no drone is opened or used (i.e., $m = 0$), where $m$ represents the number of used drones, and the tree $T$ is empty, represented by $root$ = NULL. 
	As one node represents one drone, in the remaining part of this paper, we will use `assign to a drone' and `assign to a node' as a synonym. While assigning a delivery to a node means assigning that delivery to the drone corresponds to that node. A node is called \textit{feasible} for the delivery $j$, if the node's existing assignment with $I_j$ is a feasible assignment. For assigning a delivery $j$ according to the non-decreasing launching time, the algorithm calls the function \textsc{Find}($root, j$) to find the node where we can assign the delivery. If \textsc{Find} returns a non-null pointer of a node ($node$), then we assign the delivery $j$ at $node.index$ and call \textsc{Update}(). In this function we decrease the $node.key$ by ($node.key - c_j$) and update \textit{node.data} by the maximum of previous $node.data$ and $t_j^R$.  Decreasing the $key$ of a node may violate the search property of the balanced binary search tree. Then, we delete the $node$ first and then insert a new node with this same attributes of the deleted node.
	
	Otherwise (\textsl{Find} returns a null pointer), a new node ($node$) is introduced with index $(m + 1)$ and assign the delivery $j$. Then call \textsc{Insert}() for inserting a new node in the balanced binary search tree with $node.key$ as $(B-c_j)$ and $node.data$ as $t_j^R$.
	
	\begin{algorithm}[]
		\caption{\textsc{GreedyAlgoForDDP}($I, c, B$)}\label{GreedyAlgo}
		%\State \textbf{Input:-} Set of delivery intervals $I = \{I_1, I_2, \cdots, I_n\}$; cost $c_j$ for each $I_j \in I$; drone battery budget $B$.
		\DontPrintSemicolon
		\textbf{Initialize:} $m = 0;  S_i = \phi, \forall i: 1\leq i \leq n$; $root$ = NULL\; 
		Sort the intervals according to their launching time. (without loss of generality let $t_1^L \leq t_2^L \cdots, \leq t_n^L)$\;
		\For{$j \gets 1 , n$}{
			$node$ = \textsc{Find}($root, j$)\;
			%\Comment{find the drone where the delivery $j$ can be assigned}
			\uIf{$node \neq$ NULL}{
				$i \gets$ $node.index$; $S_i \gets S_i \cup \{I_j\}$\;
				%\Comment{delivery $j$ assign to drone $i$}
				$DATA$ $\gets \max\{t_j^R,$ $node.data$\}\;
				
				\textsc{Update}($root$, $node$, $node.key$ - $c_j$, $DATA$)\;
				%\Comment{Decrease $node.key$ by $c_j$}\\ \Comment{and set $node.data$ as $t_j^R$}
			}
			\Else{
				$m \gets m+1$\;
				$S_m \gets \{I_j\}$\;
				%\Comment{a new drone is opened}
				\textsc{Insert}($root$, $m, B-c_j, t_j^R$)\;
				%\Comment{insert a new node in the tree with index $m$, } \\ \Comment{key $B-c_j$, data $t_j^R$ and left, right child as NULL}
			}
		}
		\textbf{return} $m$,  along with the assignments $\{S_1, S_2, \cdots, S_m\}$\;
	\end{algorithm}
	
	\begin{algorithm}[]
		\caption{\textsc{Find}($node, j$)}\label{Find}
		\DontPrintSemicolon
		\uIf{$node$ = NULL}{
			%\Comment{base case}
			\textbf{return} NULL\;
		}
		\Else{
			$NODE$ = \textsc{Find}($node.right, j$)\;
			%\Comment{find the node into the right subtree}
			\If{$NODE$ $\neq$ NULL}{
				\textbf{return} $NODE$\;
			}
			%\State temp $= d_j$
			\uIf{\textsc{Check}($node, j$)}{
				%\Comment{whether the delivery $j$ can be assigned to this node}
				\textbf{return} $node$\;
			}
			\ElseIf{$d_j = 0$}{
				\textbf{return} NULL\;
				%\Comment{no further \textsc{Check} to be done}
				%\ElsIf{$d_j = 0$}
				%\Comment{$d_j$ check has been done}
				%\State \textbf{return} NULL
			}
			\Else{
				$NODE$ = \textsc{Find}($node.left, j$)\;
				%\Comment{find the node into the left subtree}
				\If{$NODE$ $\neq$ NULL}{
					\textbf{return} $NODE$\;
				}
			}
		}
		\textbf{return} NULL\;
		%\Comment{\textsc{Check} for all the node rooted at $node$ has been done}
	\end{algorithm}
	
	\noindent \textit{Description of Algorithm \ref{Find}:} \textsc{Find}() takes $node$ and $j$ as the inputs. It finds the \textit{feasible node} via reverse in-order traversal in the balanced binary search tree where the delivery $j$ can be assigned. It returns the pointer of a node if the node is found in the tree rooted at $node$ is feasible for the delivery $j$ or NULL if no feasible node is found.
	
	If the tree is empty or the $node$ is NULL, the algorithm returns NULL. 
	Otherwise, the algorithm executes in three parts. It may return a non-null or null pointer at any of these parts.
	\begin{itemize}
		\item {\textsc{Find}($node.right, j$)}
		At first, algorithm finding the feasible node in the right subtree of $node$ (line $4$) by calling \textsc{Find}($node.right, j$). If it returns a non-null pointer, the algorithm returns that pointer (line $5-6$). Else, algorithm proceeds for the next part for checking whether the \textit{node} itself is feasible for the delivery $j$ by calling \textsc{Check}($node, j$).
		\item {\textsc{Check}($node, j$)} 
		If \textsl{Check} returns TRUE, then algorithm returns the pointer of the \textit{node} itself (line $7-8)$. If \textsc{Check} returns FALSE and $d_j = 0$ i.e., remaining battery capacity  of $node$ (\textit{node.key}) is lesser than the cost of the delivery $j$ $(c_j)$ (follows from line $3-5$ of \textsc{Check}), algorithm returns NULL (line $9-10$). Because all the other nodes reachable via reverse inorder traversal  have less remaining capacity, none of those nodes can be feasible for the delivery $j$. Here note that, checking for the delivery $j$ returns FALSE at most  $N(j)$ ($= d_j - 1$) times, as delivery $j$ is in conflict with $N(j)$ other intervals. Otherwise (\textsc{Check} returns FALSE and $d_j > 0$), algorithm proceeds for the next part for finding the feasible node in the left subtree of the $node$ by calling \textsc{Find}($root.left, j$).
		\item {\textsc{Find}($root.left, j$)}
		If it returns a non-null pointer, then algorithm returns that pointer (line $12-14$). 
	\end{itemize} 
	Finally, if \textsc{Find} does not return any pointer in the above three parts, it means that all nodes in the left and right subtrees of $node$ including the \textit{node} itself have been checked and found that no node in the tree rooted at $node$ is feasible for the delivery $j$. So, it returns NULL (line $15$).

	\noindent \textit{Description of Algorithm \ref{Check}}: \textsc{Check()} takes $node$, $j$ as inputs, then it returns TRUE or FALSE depending on whether delivery $j$ can be assigned to the $node$. If the $node$ has remaining battery capacity $(node.key)$ greater than or equals to the cost of the delivery $j$  $(c_j)$ and there is no conflict between the existing assignment of the $node$ and the delivery $j$, then delivery $j$ can be assigned to this $node$ and so algorithm returns TRUE. If $node.data < t_j^L$ holds, then $t_k^R < t_j^L$ also holds, $\forall$ delivery $k$ in the existing assignment of the $node$, implies delivery $j$ is compatible with the existing assignment of the $node$. On the other hand, if $node.data \geq t_j^L$ holds, then there exists  a delivery $k$ in the existing assignment of the $node$ such that $t_j^L \leq t_k^R$ holds. Also, as the deliveries are assigning according to their non-decreasing launching time, $t_k^L \leq t_j^L$ holds, implies the delivery $k$ is  in  conflict with the existing assignment of the $node$. Thus, the compatibility condition can be checked by $node.data < t_j^L$ (line $1$).
	
	If $node.key < c_j$, no further checking is required, as all the other nodes after the $node$ (in reverse inorder) have lesser or equals to remaining capacity than $c_j$. So, all these nodes are not feasible for the delivery $j$. For this case, \textsc{Check} returns FALSE and set $d_j = 0$, as no further checking is required for the delivery $j$.
	If $node.data \geq t_j^L$, then checking to be done at the predecessor of the $node$ and decrease the maximum checking number for the delivery $j (d_j)$ by one (line $7$) and the algorithm returns FALSE.
	
	Hence, algorithm \ref{Check} confirms that the assignment for each node is feasible, and so is the assignment of each used drone.
	
	\begin{algorithm}[]
		\caption{\textsc{Check}($node, j$)}\label{Check}
		\DontPrintSemicolon
		\uIf{$node.key \geq c_j$ and $node.data < t_j^L$}{
			\textbf{return} TRUE \;
			%\Comment{delivery $j$ can be assigned to this node}
		}
		\ElseIf{$node.key < c_j$}{
			$d_j \gets 0$\;
			%\Comment{No further check is needed}
			\textbf{return} FALSE\;
			%\Comment{a new drone need to introduced}
		}
		\Else{
			$d_j \gets d_j - 1$\;
			%\Comment{remaining number of \textsc{Check} for the delivery $j$}
			\textbf{return} FALSE\;
			%\Comment{checking to be done in the predecessor node}
		}
	\end{algorithm}
	
	\begin{lemma}
		\label{lemma1_factor}
		
		If $m$ is number of used (opened) drones returned by the algorithm \ref{GreedyAlgo} then, at least $(m - \Delta - 1)$ drones have used energy cost at least $\frac{B}{2}$, where $\Delta$ is the maximum degree of the interval graph $G$ constructed from delivery intervals.
		
	\end{lemma}
	\begin{proof}
		We prove the lemma by contradiction.
		
		Let at most $(m - \Delta - 2)$ used drones have used energy cost at least $\frac{B}{2}$. Then, at least $(\Delta + 2)$ drones have used energy cost < $\frac{B}{2}$, without loss of generality let $\{1, 2, \cdots, \Delta +2\}$ be those drones.
		
		Now,  consider the time ($t$), when the drone $(\Delta + 2)$ was introduced by the algorithm. Let for the delivery $j$ it was introduced. Then either $S_i^t \cap I_j \neq \phi$ or $W_i^t + c_j > B$ ($1 \leq i \leq (\Delta + 1)$) holds, where $S_i^t$ is the existing assignment and $W_i^t$ is the total cost for the existing assigned deliveries of the drone $i$ at time $t$ of the algorithm.  $I_j$ can conflict with at most $\Delta$ many $S_i^t$s, as $\Delta$ is the maximum degree of $G$, i.e., $\Delta$ is the maximum conflict number of any interval in $I$. Then, we can find an $i'$ $(1 \leq i' \leq \Delta + 1)$ for which $W_{i'}^t + c_j > B$ holds. But, at the end of the algorithm (say at time $t_e$), $W_{i'}^{t_e} < \frac{B}{2}$ and $W_{\Delta+2}^{t_e} < \frac{B}{2}$, implies $c_j < \frac{B}{2}$, implies $W_{i'}^{t_e} + c_j < B$, implies the drone $i'$ is feasible the delivery $j$, which is a contradiction.
		Hence, statement of the lemma follows.  
	\end{proof}
	\begin{lemma}
		\label{lemma1_compl}
		The time complexity for algorithm \ref{GreedyAlgo} is $\mathcal{O}(n\log n + n_e)$, where $n$ is the number of deliveries and $n_e$ is the number of edges in $G$.
	\end{lemma}
	\begin{proof}
		At first, algorithm sorts the intervals according to their launch time, takes $\mathcal{O}(n \log n)$ time.
		
		Consider the time of execution, when the algorithm wants to assign the delivery $j$. At this point, number of drones is opened is at most $(j - 1)$, so as number of nodes in the tree. Tree is being balanced binary search tree, at this time height of the tree is $\mathcal{O}(\log j$). For finding the appropriate node in the tree, which is feasible for the delivery $j$, algorithm checks at most $d_j = (N(j) + 1)$ maximum nodes (according to their key) in the tree by calling the function \textsl{Check}(). We can find all these $d_j$ nodes in $\mathcal{O}(\log j) + \mathcal{O}(N(j))$ steps by recursively calling \textsl{Find}() from the $root$. For each \textsl{Check}() algorithm needs $\mathcal{O}(1)$ time. So, total time complexity for \textsl{Find}() for the delivery $j$ is $\mathcal{O}(\log j) + \mathcal{O}(N(j))$. After this, algorithm needs either an \textsl{Insert}() or an \textsl{Update}() operation to the tree, which takes $\mathcal{O}(\log j)$ time.
		
		So, for a delivery $j$, finding, checking, assigning and inserting or updating, takes $\mathcal{O}(\log j) + \mathcal{O}(N(j))$ time. Therefore, for assigning all the deliveries, algorithm needs $\sum_{j=1}^n \mathcal{O}(\log j) + \mathcal{O}(N(j)) = \mathcal{O}(n\log n) + \mathcal{O}(n_e)$, where $n_e = \frac{1}{2} \sum_{j=1}^n N(j)$ is the number of edges in $G$.
		
		Thus, the total running time for the algorithm is $\mathcal{O}(n\log n + n_e)$.  
	\end{proof}
	\begin{theorem}
		\label{Theorem-2}
		Algorithm \ref{GreedyAlgo} is an approximation algorithm for DDP, uses at most $2 OPT + (\Delta + 1)$ drones, where $n$ is the number of deliveries, OPT is the optimum number of drones required for DDP,  and $\Delta$ is the maximum degree of the interval graph $G$.
	\end{theorem}
	\begin{proof}
		At the beginning, the algorithm sorts the deliveries based on their non-decreasing launch time. Without loss of generality let $t_1^L \leq t_2^L \cdots, \leq t_n^L$. Then, it assigns all deliveries in accordance with the previously sorted order. First, algorithm takes the delivery with shortest launching time and assign to a drone (indexed by $1$). We call this drone as a used drone. Since cost for each delivery is $\leq B$, this is an feasible assignment. Thereafter, for assigning the delivery $I_j$ $(2 \leq j \leq n)$ as per their sorted order, algorithm checks whether there exists a drone $i$ among all the used drones (say $m$), which is feasible for this delivery. If yes, assign the delivery to the drone $i$. Otherwise, a new drone is introduced with index ($m + 1$) and assign the delivery $I_j$ there. The feasibility of each existing assignment is always upheld for each such addition of the delivery. Thus, at the end of the algorithm \ref{GreedyAlgo}, all the returned assignments are feasible.  
		%The feasibility of the solution returns by the algorithm \ref{GreedyAlgo} is directly follows from the descriptions of algorithm \ref{GreedyAlgo}, \ref{Find}, \ref{Check}, \ref{Insert}, \ref{Update}.
		
		The polynomial running time of the algorithm follows from lemma \ref{lemma1_compl}.
		
		If $m$ is the number of drones (used drones) returned by the algorithm, then from the lemma \ref{lemma1_factor} following holds.
		\begin{alignat}{2}
			&(m - \Delta -1).\frac{B}{2} \leq \sum_{j=1}^n c_j \leq OPT.B\\
			&\Rightarrow m \leq 2.OPT + (\Delta + 1)
		\end{alignat}
		Hence the proof.  
	\end{proof}
	
	\subsection{Approximation Algorithm for DDP using Colouring}
	
	In this section, we demonstrate another approximation algorithm (Algorithm \ref{Algo-2}) using colouring to solve DDP. The algorithm in general gives a  better solution than the previous one. 
	
	%For the given set of delivery time intervals $I$, we can construct an interval graph $G$, where the vertices represent the intervals and two vertices are adjacent if the corresponding two intervals conflict. Now $G$ being an interval graph, it is \textit{perfect} \cite{west}. Therefore, $\chi(G)$ = $\omega(G)$, where $\chi(G)$ is the \textit{chromatic number} and $\omega(G)$ is the \textit{maximum clique size} of $G$.
	%We use $\omega$ instead of $\omega(G)$ for simplicity. So, the vertices of $G$ can be optimally coloured by using $w$ many colours in polynomial time (linear in terms of number of vertices and edges) \cite{chordal} such that no two adjacent vertices get the same colour. Since every interval in $I$ represented by a unique vertex in $G$, each interval can be coloured linearly from the colour of $G$.  Further, we can partition $I$ into $\omega$ many conflict-free (i.e., compatible) sets corresponding to each colour. Let $J_1, J_2, \cdots, J_{\omega}$ be the partitions. Therefore, $\cup_{k=1}^{\omega} J_k = I$.
	
	From the given interval set $I$, we can construct an interval graph $G$ as described in section \ref{section:model}. 
	Interval graph being a perfect graph, vertices of $G$ can be coloured by $\omega$ many colours, where $\omega$ is the maximum clique size of $G$.  From the $\omega$-colouring of $G$, we can partition the set $I$ into $\omega$ many sets, each corresponding to the same coloured vertices of $G$. 
	Let $\{J_1, J_2, \cdots, J_{\omega}\}$ be the  partition set, where $J_k \subseteq I$ is a compatible set associated to colour $k$ $(1 \leq k \leq \omega)$.
	
	For each $k$ $(1 \leq k \leq \omega)$, we can find the number of drones needed to schedule all the intervals in $J_k$ and the corresponding schedules using algorithm \ref{GreedyAlgo_mod}, a modified version of algorithm \ref{GreedyAlgo}. Similar to the algorithm \ref{GreedyAlgo}, for each interval set $J_k$, we construct a balanced binary search tree $T_k$ $(1 \leq k \leq \omega)$. 
	
	The pseudocode of the modified version of algorithm \ref{GreedyAlgo} is depicted in algorithm \ref{GreedyAlgo_mod}.
	Following changes are made on algorithm \ref{GreedyAlgo} to modify it to algorithm \ref{GreedyAlgo_mod}. 
	As $J_k$ is compatible, then the $data$ from each tree node can be omitted. All the other attributes ($node.key, node.index, node.left, node.right$) of a node in each tree remains same. For a fixed interval set $J_k$ $(1 \leq k \leq \omega)$, the index of $i-th$ used drone is denoted by $i_k$ and the corresponding assignment is denoted by $S_{i_k}$. For assigning a delivery $I_j$ in $J_k$, algorithm \ref{GreedyAlgo_mod} calls \textsc{FindModified}($root, c_j$). If \textsc{FindModified}($root, c_j$) returns NULL, a new drone is introduced with index $(m+1)_k$ (initially, $m = 0$) and assign the delivery $I_j$ for the drone. Then the algorithm calls \textsc{InsertModified}() for inserting a new node with index $(m+1)_k$ and key as $(B - c_j)$. If \textsc{FindModified}($root, c_j$) returns a non-null pointer of a $node$ then the algorithm assigns the delivery $I_j$ at the $node$ and calls \textsc{UpdateModified}() for decreasing the $node.key$ by $c_j$. Whereas \textsc{FindModified}() finds the $node$ with maximum $key$ in the tree. If the $node$ with maximum key is not feasible for the delivery $I_j$, then no other node in the tree will be feasible for $I_j$. So, for this case, \textsc{FindModified} returns NULL. Otherwise, \textsc{FindModified} returns the pointer of the $node$.

	\begin{algorithm}[]
		\caption{\textsc{GreedyAlgoModified}($J_k, c, B$)}\label{GreedyAlgo_mod}
		\DontPrintSemicolon
		%\State \textbf{Input:-} Set of delivery intervals $I = \{I_1, I_2, \cdots, I_n\}$; cost $c_j$ for each $I_j \in I$; drone battery budget $B$.
		\textbf{Initialize:} $m = 0;  S_{i_k} = \phi, \forall i: 1\leq i \leq n$; $root$ = NULL\; 
		\For{each interval $I_j$ in $J_k$}{
			$node$ = \textsc{FindModified}($root, c_j$)\;
			%\Comment{find the feasible node for the delivery $j$}
			\If{$node \neq$ NULL}{
				$i_k \gets$ $node.index$; $S_{i_k} \gets S_{i_k} \cup \{I_j\}$\;
				%\Comment{delivery $j$ assign to the drone $i_k$}
				\textsc{UpdateModified}($root$, $node$, $node.key$ - $c_j$)\;
				%\Comment{Decrease $node.key$ by $c_j$}
			}
			\Else{
				$m \gets m+1$\;
				$S_{m_k} \gets \{I_j\}$\;
				%\Comment{a new drone is opened}
				\textsc{InsertModified}($root$, $m_k, B-  c_j$)\;
				%\Comment{a new node inserted in the tree with } \\ %\Comment{index $m_k$, key $B - c_j$, and left, right child as NULL}
			}
		}
		\textbf{return} $m_k$,  with their assignments $S_{i_k}$\;
	\end{algorithm}
	
	\begin{algorithm}[]
		\caption{\textsc{FindModified}($node, c_j$)}\label{Find_mod}
		\DontPrintSemicolon
		\uIf{$node$ = NULL}{
			\textbf{return} NULL \;
			%\Comment{base case}
		}
		\Else{
			\uIf{$node.right = NULL$ and $node.key \geq c_j$}{
				\textbf{return} $node$\;
			}
			\ElseIf{$node.right = NULL$ and $node.key < c_j$}{
				\textbf{return} NULL\;
			}
			\Else{
				\textbf{return} \textsc{FindModified}($node.right, c_j$)\;
			}
			%\EndIf
		}
	\end{algorithm}
	Let $m_k$ be the number of drones returned by the algorithm \ref{GreedyAlgo_mod} for the delivery interval set $J_k$. Then, the following lemma \ref{lemma-color} is true.
	\begin{lemma}
		\label{lemma-color}
		$\mathcal{W}(J_k) >  \left(\frac{m_k - 1}{2}\right). B$, where $m_k$ is the number of drones returned by algorithm \ref{GreedyAlgo_mod} for the interval set $J_k \subseteq I$ and $\mathcal{W}(J_k) = \sum_{I_j \in J_k} (c_j)$.
	\end{lemma}
	\begin{proof}
		Let $m_k$ drones are denoted by $\{1, 2, \cdots, m_k\}$ and $S_i$ = \{set of intervals in $J_k$ which are assigned to the drone $i$ by the algorithm \ref{GreedyAlgo_mod}\} $(1 \leq i \leq m_k)$.\\
		Then, $\mathcal{W}(S_{2i - 1}) + \mathcal{W}(S_{2i}) > B$ for $(1 \leq i \leq \lfloor \frac{m_k}{2} \rfloor)$, otherwise all the intervals in $S_{2i-1}$ and $S_{2i}$ can be assigned to a single drone. Thus,
		\begin{alignat}{2}
			&\sum_{i=1}^{\lfloor \frac{m_k}{2} \rfloor} (\mathcal{W}(S_{2i - 1}) + \mathcal{W}(S_{2i})) > \left\lfloor \frac{m_k}{2} \right\rfloor . B \\
			&\Rightarrow \mathcal{W}(J_k) > \left\lfloor \frac{m_k}{2} \right\rfloor . B\\
			&\Rightarrow \mathcal{W}(J_k) >  \left(\frac{m_k - 1}{2}\right). B
		\end{alignat}  
	\end{proof}
	
	%Now we concretely define the pseudocode for the algorithm in Algorithm \ref{Algo-2}.
	%We de the pseudocode for the algorithm in Algorithm \ref{Algo-2}.
	
	\begin{algorithm}[]
		\caption{\textsc{ApprxAlgoWithColouring}}\label{Algo-2}
		\DontPrintSemicolon
		\textbf{Input:} Set of delivery intervals $I = \{I_1, I_2, \cdots, I_n\}$; cost $c_j$ for each $I_j \in I$; drone battery budget $B$\;
		Construct an interval graph $G$ from the delivery time interval set $I$\;
		Find maximum clique size $(\omega)$ of $G$\;
		Colour all the vertices of $G$ with the colors $\{1, 2, \cdots, \omega\}$ such that no two adjacent vertices gets same colour\;
		Find $J_k$ = \{Set of intervals in $G$ whose corresponding vertices in $G$ are coloured by the colour $k$\} $(1 \leq k \leq \omega)$\;
		For each $J_k$ $(1 \leq k \leq \omega)$, find number of drones, say $m_k$ and corresponding assignments, say $\mathcal{S}_k$ by using the algorithm \ref{GreedyAlgo_mod}\;
		Return $\sum_{k=1}^{\omega} m_k$ along with their corresponding assignments\;
	\end{algorithm}
	
	\begin{lemma}
		\label{lemma2_compl}
		The time complexity for Algorithm \ref{Algo-2} is $\mathcal{O}(n\log n + n_e)$, where $n$ is the number of deliveries and $n_e$ is the number of edges in $G$.
	\end{lemma}
	\begin{proof}
		We can construct an interval graph $G$ (line 2) from the given interval set $I$ in $\mathcal{O}(n + n_e)$ time. Then, find maximum clique size $\omega$ (line 3) and the colour all the vertices of $G$ (line 4) with these $\omega$ colours can be done in $\mathcal{O}(n
		+ n_e)$ time \cite{chordal}. Finding all the $J_k$ $(1 \leq k \leq \omega)$  (line $5$) takes $\mathcal{O}(n)$ time.
		
		Algorithm \ref{Algo-2} uses algorithm \ref{GreedyAlgo_mod} for finding the number of drones and corresponding assignments for each delivery intervat $J_k$ $(1 \leq k \leq \omega)$ (line $6$). For each interval set $J_k$ $(1 \leq k \leq \omega)$ we create a balanced binary search tree $(T_k)$ similar to the previous section. For assigning a delivery $I_j$ in $J_k$, algorithm \ref{GreedyAlgo_mod} call \textsc{FindModified}($root, c_j)$. This returns a pointer in $\mathcal{O}(h_j^k)$ time, where $h_j^k$ is the height of the tree $T_k$ before the assignment of delivery $I_j$. If \textsc{FindModified}($root, C_j)$ returns NULL, the algorithm calls \textsc{Insert}(). Otherwise, the algorithm  calls \textsc{Update}(). For any of the case, algorithm \ref{GreedyAlgo_mod} needs $\mathcal{O}(h_j^k)$ time for assigning delivery $I_j$ in $J_k$. Thus, for the interval set $J_k$ algorithm \ref{GreedyAlgo_mod} runs in $\sum_{I_j \in J_k} \mathcal{O}(h_j^k) \leq \mathcal{O}(n_k \log n_k)$ time, where $n_k$ is the number of deliveries in $J_k$. So, total running time for line $6$ of algorithm \ref{Algo-2} is $\sum_{k=1}^{\omega} \mathcal{O}(n_k \log n_k) = \mathcal{O}(n \log n)$.
		
		Hence, overall running time for Algorithm \ref{Algo-2} is $\mathcal{O}(n \log n + n_e)$.  
	\end{proof}
	
	\begin{theorem}
		\label{Theorem-3}
		Algorithm \ref{Algo-2} is an approximation algorithm for DDP, which uses at most ($2 OPT + \omega$) drones, where $\omega$ is the maximum clique size of $G$ and OPT is optimum number of drones required for DDP.
	\end{theorem}
	\begin{proof}
		Algorithm \ref{Algo-2} uses algorithm \ref{GreedyAlgo_mod} as it's subroutine for each delivery interval set $J_k$ ($1 \leq k \leq \omega$). Algorithm \ref{GreedyAlgo_mod} is a modified version of algorithm \ref{GreedyAlgo} and from theorem \ref{Theorem-2}, we know that algorithm \ref{GreedyAlgo} always gives feasible solution. So, $\{J_1, J_2, \cdots, J_{\omega}\}$ being partition of the given delivery interval set $I$, algorithm \ref{Algo-2} gives feasible assignments for each delivery in $I$.
		
		Polynomial running time of the algorithm follows from lemma \ref{lemma2_compl}.
		
		Let $\mathcal{W}(I)$ be the total cost of all the deliveries in $I$ and $\mathcal{W}(J_k)$ be the total cost of all the deliveries in $J_k$, $\forall k: 1 \leq k \leq \omega$. 
		
		Then, $\sum_{k=1}^{\omega} \mathcal{W}(J_k) = \mathcal{W}(I) \leq OPT.B$.
		
		Let $m_k$ be the number of drones returned by the algorithm \ref{GreedyAlgo_mod} for the interval set $J_k$, $\forall k: 1 \leq k \leq \omega$, Thus, from the lemma \ref{lemma-color} following holds.
		\begin{alignat}{2}
			&\sum_{i=1}^{\omega} \left(\frac{m_k - 1}{2}\right). B < \sum_{k=1}^{\omega} \mathcal{W}(J_k) \\
			&\sum_{i=1}^{\omega} \left(\frac{m_k - 1}{2}\right). B < OPT.B \\
			&\Rightarrow \sum_{i=1}^{\omega} m_k < 2.OPT + \omega. \label{eq-14}
		\end{alignat}
		Hence the proof.  
	\end{proof}
	\section{Conclusion}
	\label{section:concl}
	% In this paper, we studied drone-delivery packing problem (DDP). We propose two approximation algorithm with identical running time $\mathcal{O}(n\log n) + \mathcal{O}(n_e)$, where $n$ is the number of deliveries and $n_e$ is the number of edges in the interval graph, constructed from the delivery time intervals. First akgorithm uses ($2OPT + \Delta + 1$) drones and second algorithm uses ($2OPT + \omega$) drones, where $\omega$ is the maximum clique size, $\Delta$ is the maximum degree of the the interval graph and OPT is optimum number of drones required for given DDP. In general $\omega \leq (\Delta + 1)$, so the second algorithm gives better approximation algorithm than the first. Finding effective constant factor approximation algorithms, in addition asymptotic polynomial time approximation scheme (PTAS) for DDP, will be consider for future research. Furthermore, if the drone has a charging area inside the truck, determining the delivery schedule using the fewest possible drones subject to a finite amount of charging time is of great interest.
	
	In this paper, we studied the drone-delivery packing problem (DDP). We propose two approximation algorithms with identical running time $\mathcal{O}(n\log n + n_e)$, where $n$ is the number of deliveries and $n_e$ is the number of edges in the interval graph $G$, constructed from the delivery time intervals. The first algorithm  (Algorithm \ref{GreedyAlgo}) uses $2OPT + (\Delta + 1$) drones, and the second algorithm (Algorithm \ref{Algo-2})  uses $2OPT + \omega$ drones, where $\omega$ is the maximum clique size of $G$, $\Delta$ is the maximum degree of $G$ and OPT is the optimum number of drones required for the DDP. In general, $\omega \leq (\Delta + 1)$, so, second algorithm gives a better approximation than the first. 
	However, the second algorithm gives $3$-factor approximation as we need at least $\omega$ many drones for scheduling all the deliveries, i.e., $\omega \leq OPT$. 
	
	Finding better constant factor approximation algorithms and asymptotic polynomial time approximation schemes (PTAS) for DDP will be considered for future research. Furthermore, if the drone has a charging area inside the truck, determining the delivery schedule using the fewest possible drones subject to a finite amount of charging time is of great interest.
	\bibliographystyle{ACM-Reference-Format}
	\balance
	\bibliography{arxiv}


\begin{thebibliography}{17}

%%% the Web address.

\ifx \showCODEN    \undefined \def \showCODEN     #1{\unskip}     \fi
\ifx \showDOI      \undefined \def \showDOI       #1{#1}\fi
\ifx \showISBNx    \undefined \def \showISBNx     #1{\unskip}     \fi
\ifx \showISBNxiii \undefined \def \showISBNxiii  #1{\unskip}     \fi
\ifx \showISSN     \undefined \def \showISSN      #1{\unskip}     \fi
\ifx \showLCCN     \undefined \def \showLCCN      #1{\unskip}     \fi
\ifx \shownote     \undefined \def \shownote      #1{#1}          \fi
\ifx \showarticletitle \undefined \def \showarticletitle #1{#1}   \fi
\ifx \showURL      \undefined \def \showURL       {\relax}        \fi
% invisible to TeX
\providecommand\bibfield[2]{#2}
\providecommand\bibinfo[2]{#2}
\providecommand\natexlab[1]{#1}
\providecommand\showeprint[2][]{arXiv:#2}

\bibitem[Amazon({[n.\,d.]})]%
        {amazon}
\bibfield{author}{\bibinfo{person}{Amazon}.}
  \bibinfo{year}{[n.\,d.]}\natexlab{}.
\newblock \bibinfo{title}{Amazon customers in Lockeford, California, will be
  among the first to receive Prime Air drone deliveries in the U.S.}
\newblock
\newblock
\urldef\tempurl%
\url{https://www.aboutamazon.com/news/transportation/amazon-prime-air-prepares-for-drone-deliveries}
\showURL{%
\tempurl}


\bibitem[Betti~Sorbelli et~al\mbox{.}(2022)]%
        {Betti_ICDCN22}
\bibfield{author}{\bibinfo{person}{Francesco Betti~Sorbelli},
  \bibinfo{person}{Federico Cor\`{o}}, \bibinfo{person}{Sajal~K. Das},
  \bibinfo{person}{Lorenzo Palazzetti}, {and} \bibinfo{person}{Cristina~M.
  Pinotti}.} \bibinfo{year}{2022}\natexlab{}.
\newblock \showarticletitle{Greedy Algorithms for Scheduling Package Delivery
  with Multiple Drones}. In \bibinfo{booktitle}{\emph{23rd International
  Conference on Distributed Computing and Networking}} (Delhi, AA, India)
  \emph{(\bibinfo{series}{ICDCN 2022})}. \bibinfo{publisher}{Association for
  Computing Machinery}, \bibinfo{address}{New York, NY, USA},
  \bibinfo{pages}{31–39}.
\newblock
\showISBNx{9781450395601}
\urldef\tempurl%
\url{https://doi.org/10.1145/3491003.3491028}
\showDOI{\tempurl}


\bibitem[Boysen et~al\mbox{.}(2018)]%
        {Boysen2018DroneDF}
\bibfield{author}{\bibinfo{person}{Nils Boysen}, \bibinfo{person}{Dirk
  Briskorn}, \bibinfo{person}{Stefan Fedtke}, {and} \bibinfo{person}{Stefan
  Schwerdfeger}.} \bibinfo{year}{2018}\natexlab{}.
\newblock \showarticletitle{Drone delivery from trucks: Drone scheduling for
  given truck routes}.
\newblock \bibinfo{journal}{\emph{Networks}}  \bibinfo{volume}{72}
  (\bibinfo{year}{2018}), \bibinfo{pages}{506 -- 527}.
\newblock


\bibitem[Boysen et~al\mbox{.}(2021)]%
        {LastMile}
\bibfield{author}{\bibinfo{person}{Nils Boysen}, \bibinfo{person}{Stefan
  Fedtke}, {and} \bibinfo{person}{Stefan Schwerdfeger}.}
  \bibinfo{year}{2021}\natexlab{}.
\newblock \showarticletitle{Last-mile delivery concepts: a survey from an
  operational research perspective}.
\newblock \bibinfo{journal}{\emph{OR Spectrum}}  \bibinfo{volume}{43}
  (\bibinfo{date}{03} \bibinfo{year}{2021}), \bibinfo{pages}{1--58}.
\newblock
\urldef\tempurl%
\url{https://doi.org/10.1007/s00291-020-00607-8}
\showDOI{\tempurl}


\bibitem[Coffman et~al\mbox{.}(1984)]%
        {Coffman1984}
\bibfield{author}{\bibinfo{person}{E.~G. Coffman}, \bibinfo{person}{M.~R.
  Garey}, {and} \bibinfo{person}{D.~S. Johnson}.}
  \bibinfo{year}{1984}\natexlab{}.
\newblock \bibinfo{booktitle}{\emph{Approximation Algorithms for Bin-Packing
  --- An Updated Survey}}.
\newblock \bibinfo{publisher}{Springer Vienna}, \bibinfo{address}{Vienna},
  \bibinfo{pages}{49--106}.
\newblock
\showISBNx{978-3-7091-4338-4}
\urldef\tempurl%
\url{https://doi.org/10.1007/978-3-7091-4338-4_3}
\showDOI{\tempurl}


\bibitem[Coffman~Jr. et~al\mbox{.}(2013)]%
        {bin_survey}
\bibfield{author}{\bibinfo{person}{Edward~G. Coffman~Jr.},
  \bibinfo{person}{J{\'a}nos Csirik}, \bibinfo{person}{G{\'a}bor Galambos},
  \bibinfo{person}{Silvano Martello}, {and} \bibinfo{person}{Daniele Vigo}.}
  \bibinfo{year}{2013}\natexlab{}.
\newblock \bibinfo{booktitle}{\emph{Bin Packing Approximation Algorithms:
  Survey and Classification}}.
\newblock \bibinfo{publisher}{Springer New York}, \bibinfo{address}{New York,
  NY}, \bibinfo{pages}{455--531}.
\newblock
\showISBNx{978-1-4419-7997-1}
\urldef\tempurl%
\url{https://doi.org/10.1007/978-1-4419-7997-1_35}
\showDOI{\tempurl}


\bibitem[Crişan and Nechita(2019)]%
        {CRISAN201938}
\bibfield{author}{\bibinfo{person}{Gloria~Cerasela Crişan} {and}
  \bibinfo{person}{Elena Nechita}.} \bibinfo{year}{2019}\natexlab{}.
\newblock \showarticletitle{On a cooperative truck-and-drone delivery system}.
\newblock \bibinfo{journal}{\emph{Procedia Computer Science}}
  \bibinfo{volume}{159} (\bibinfo{year}{2019}), \bibinfo{pages}{38--47}.
\newblock
\showISSN{1877-0509}
\urldef\tempurl%
\url{https://doi.org/10.1016/j.procs.2019.09.158}
\showDOI{\tempurl}


\bibitem[Daknama and Kraus(2017)]%
        {Daknama2017VehicleRW}
\bibfield{author}{\bibinfo{person}{Rami Daknama} {and}
  \bibinfo{person}{Elisabeth Kraus}.} \bibinfo{year}{2017}\natexlab{}.
\newblock \showarticletitle{Vehicle Routing with Drones}.
\newblock \bibinfo{journal}{\emph{ArXiv}}  \bibinfo{volume}{abs/1705.06431}
  (\bibinfo{year}{2017}).
\newblock


\bibitem[Dutta and Goswami(2020)]%
        {agriculture}
\bibfield{author}{\bibinfo{person}{Gopal Dutta} {and} \bibinfo{person}{Purba
  Goswami}.} \bibinfo{year}{2020}\natexlab{}.
\newblock \showarticletitle{Application of drone in agriculture: A review}.
\newblock \bibinfo{journal}{\emph{International Journal of Chemical Studies}}
  \bibinfo{volume}{8} (\bibinfo{date}{10} \bibinfo{year}{2020}),
  \bibinfo{pages}{181--187}.
\newblock
\urldef\tempurl%
\url{https://doi.org/10.22271/chemi.2020.v8.i5d.10529}
\showDOI{\tempurl}


\bibitem[Garey and Johnson(1979)]%
        {Np}
\bibfield{author}{\bibinfo{person}{Michael~R. Garey} {and}
  \bibinfo{person}{David~S. Johnson}.} \bibinfo{year}{1979}\natexlab{}.
\newblock \bibinfo{booktitle}{\emph{Computers and intractability}}.
\newblock \bibinfo{publisher}{W. H. Freeman and Co., San Francisco, Calif.}
  x+338 pages.
\newblock
\showISBNx{0-7167-1045-5}
\newblock
\shownote{A guide to the theory of NP-completeness}.


\bibitem[Kardasz and Doskocz(2016)]%
        {defence}
\bibfield{author}{\bibinfo{person}{Piotr Kardasz} {and} \bibinfo{person}{Jacek
  Doskocz}.} \bibinfo{year}{2016}\natexlab{}.
\newblock \showarticletitle{Drones and Possibilities of Their Using}.
\newblock \bibinfo{journal}{\emph{Journal of Civil \& Environmental
  Engineering}}  \bibinfo{volume}{6} (\bibinfo{date}{01} \bibinfo{year}{2016}).
\newblock
\urldef\tempurl%
\url{https://doi.org/10.4172/2165-784X.1000233}
\showDOI{\tempurl}


\bibitem[Mathew et~al\mbox{.}(2015)]%
        {HDP}
\bibfield{author}{\bibinfo{person}{Neil Mathew}, \bibinfo{person}{Stephen~L.
  Smith}, {and} \bibinfo{person}{Steven~L. Waslander}.}
  \bibinfo{year}{2015}\natexlab{}.
\newblock \showarticletitle{Planning Paths for Package Delivery in
  Heterogeneous Multirobot Teams}.
\newblock \bibinfo{journal}{\emph{IEEE Transactions on Automation Science and
  Engineering}} \bibinfo{volume}{12}, \bibinfo{number}{4}
  (\bibinfo{year}{2015}), \bibinfo{pages}{1298--1308}.
\newblock
\urldef\tempurl%
\url{https://doi.org/10.1109/TASE.2015.2461213}
\showDOI{\tempurl}


\bibitem[Murray and Chu(2015)]%
        {MURRAY201586}
\bibfield{author}{\bibinfo{person}{Chase~C. Murray} {and}
  \bibinfo{person}{Amanda~G. Chu}.} \bibinfo{year}{2015}\natexlab{}.
\newblock \showarticletitle{The flying sidekick traveling salesman problem:
  Optimization of drone-assisted parcel delivery}.
\newblock \bibinfo{journal}{\emph{Transportation Research Part C: Emerging
  Technologies}}  \bibinfo{volume}{54} (\bibinfo{year}{2015}),
  \bibinfo{pages}{86--109}.
\newblock
\showISSN{0968-090X}
\urldef\tempurl%
\url{https://doi.org/10.1016/j.trc.2015.03.005}
\showDOI{\tempurl}


\bibitem[Murray and Raj(2020)]%
        {Murray2019TheMF}
\bibfield{author}{\bibinfo{person}{Chase~C. Murray} {and}
  \bibinfo{person}{Ritwik Raj}.} \bibinfo{year}{2020}\natexlab{}.
\newblock \showarticletitle{The multiple flying sidekicks traveling salesman
  problem: Parcel delivery with multiple drones}.
\newblock \bibinfo{journal}{\emph{Transportation Research Part C: Emerging
  Technologies}}  \bibinfo{volume}{110} (\bibinfo{year}{2020}),
  \bibinfo{pages}{368--398}.
\newblock
\showISSN{0968-090X}
\urldef\tempurl%
\url{https://doi.org/10.1016/j.trc.2019.11.003}
\showDOI{\tempurl}


\bibitem[Park et~al\mbox{.}(2022)]%
        {healthcare}
\bibfield{author}{\bibinfo{person}{Hyung~Jin Park}, \bibinfo{person}{Reza
  Mirjalili}, \bibinfo{person}{Murray~J. C\^{o}t\'{e}}, {and}
  \bibinfo{person}{Gino~J. Lim}.} \bibinfo{year}{2022}\natexlab{}.
\newblock \showarticletitle{Scheduling Diagnostic Testing Kit Deliveries with
  the Mothership and Drone Routing Problem}.
\newblock \bibinfo{journal}{\emph{J. Intell. Robotics Syst.}}
  \bibinfo{volume}{105}, \bibinfo{number}{2} (\bibinfo{date}{jun}
  \bibinfo{year}{2022}), \bibinfo{numpages}{19}~pages.
\newblock
\showISSN{0921-0296}
\urldef\tempurl%
\url{https://doi.org/10.1007/s10846-022-01632-1}
\showDOI{\tempurl}


\bibitem[Stacho(2008)]%
        {chordal}
\bibfield{author}{\bibinfo{person}{Juraj Stacho}.}
  \bibinfo{year}{2008}\natexlab{}.
\newblock \emph{\bibinfo{title}{Complexity of Generalized Colourings of Chordal
  Graphs}}.
\newblock \bibinfo{thesistype}{Ph.\,D. Dissertation}. \bibinfo{address}{CAN}.
\newblock
\showISBNx{9780494468265}
\newblock
\shownote{AAINR46826}.


\bibitem[West(1996)]%
        {west}
\bibfield{author}{\bibinfo{person}{D.B. West}.}
  \bibinfo{year}{1996}\natexlab{}.
\newblock \bibinfo{booktitle}{\emph{Introduction to Graph Theory}}.
\newblock \bibinfo{publisher}{Prentice Hall}.
\newblock
\showISBNx{9780132278287}
\showLCCN{95024773}


\end{thebibliography}
	
\end{document}